\documentclass[preprint,pre,floats,showkeys,aps,amsmath,amssymb,12pt]{revtex4}
\usepackage{graphicx}
\usepackage{natbib}
\usepackage{amsthm}
\usepackage[left = 1in, top=1.2in,right=1in,bottom=1.2in,a4paper]{geometry}
\usepackage{amsmath}
\usepackage{amssymb}
\usepackage{hyperref}
\newcommand{\ket}[1]{|#1\rangle}
\newcommand{\bra}[1]{\langle #1|}
\begin{document}

\title{Does Considering Quantum Correlations Resolve the Information Paradox?}
\author{Avik Roy\footnote{avikroy@eee.buet.ac.bd}, Moinul Hossain Rahat\footnote{rahat.moin90@gmail.com}, Md.\! Abdul Matin (deceased) and Mishkat Al Alvi\footnote{mishkat.alvi@gmail.com}}
\affiliation{Department of Electrical \& Electronic Engineering, Bangladesh University of Engineering \& Technology, Dhaka 1000, Bangladesh}

\begin{abstract}
In this paper, we analyze whether quantum correlations between successive steps of evaporation can open any way to resolve the black hole information paradox. Recently a celebrated result in literature shows that `small' correction to leading order Hawking analysis fails to restore unitarity in black hole evaporation. We study a toy qubit model of evaporation allowing small quantum correlations between successive steps and verify the previous result. Then we generalize the concept of correction to Hawking state by relaxing the `smallness' condition. Our result generates a nontrivial upper and lower bound on change in entanglement entropy in the evaporation process. This gives us a quantitative measure of correction that would mathematically facilitate restoration of unitarity in black hole evaporation. We then investigate whether this result is compatible to the established physical constraints of unitary evolution of a state in a subsystem. We find that the generalized bound on entanglement entropy leads to significant deviation from Page curve. This leads us to agree with the recent claim in literature that no amount of correction in the form of Bell pair states would lead to any resolution to the information paradox.
\end{abstract}

\maketitle

\section{Introduction}


After Bekenstein had argued in favor of associating entropy with black holes \cite{bek-1}, a series of results \cite{penr}, \cite{christ}, \cite{hawk-1} converged to a complete formulation of black hole thermodynamics as presented in \cite{b-c-h}. Following Hawking's results \cite{hawk-1}, Bekenstein proposed a linear relationship between the black hole entropy and the surface area of black hole event horizon \cite{bek-2} and formulated the generalized second law of thermodynamics. Thermodynamic features of a black hole suggested that black holes should radiate like any black body, though the idea is in clear contradiction with the classical description of a black hole. In solving the conundrum, Hawking introduced a semiclassical approach introducing quantum mechanical effects on the matter field while still ignoring the quantum gravitational effects. His analysis in \cite{hawk-2} yielded a bizarre result that a black hole radiates as a black body of temperature $\frac{\hbar\kappa}{2\pi k}$. This phenomenon, called `Hawking radiation', is attributed to the vacuum fluctuations near the event horizon causing a flux of thermal radiation to be observed by an asymptotic observer. 


However, such a process violates the underlying principle of quantum evolutions -- unitarity. Reasonably, it can be assumed that a black hole is formed from some matter in pure state. On the other hand, radiation is thermal, characterized by the black hole temperature and described by a mixed state \cite{hawk-3}, \cite{gidd-1}. Two black holes with identical temperature would radiate identically irrespective of the matter that forms it, causing irretrievable loss of information. Violation of unitarity by Hawking's semiclassical results is known as the \emph{black hole information paradox}.

A reasonable way to look at the paradox is through monotonic increase of entanglement entropy. Entangled pairs are created in every step of evaporation. Hence, any model of evaporation that respects unitarity must find a way to revert the direction of change in entanglement entropy at some stage of evaporation.

A popular idea about resolving the information paradox is to allow `small' corrections to the leading order Hawking analysis. Such corrections are usually attributed to quantum correlations among radiated quanta. Although such correlations should be `small', such effects might add up to form a significant deviation from the standard thermal spectrum and resolve the paradox. By the end of its lifetime, a black hole will have radiated away most of the information \cite{presk}. However, this speculation has been invalidated in \cite{peda} using a mathematical framework of two dimensional Hilbert space spanned by two Bell pair states. This result has been exemplified by some toy models in \cite{infall} and \cite{mathur-plumberg} considering `small' correction to the Hawking state in the form of weak correlation between quanta produced in successive steps of black hole radiation.

In this paper, we study a more generalized toy model allowing arbitrary, small corrections in the form of quantum correlations between successive pairs in evaporation process and reassure the same result.
 

Can entanglement entropy start to decrease at some point of evaporation if we allow `any' amount of correction to the leading order Hawking state? In the next part of this work, we relax the `smallness' condition as prescribed in \cite{peda} and \cite{avery} and find rigorous nontrivial upper and lower bound on change in entanglement entropy. Such formulation allows us to quantify the amount of correction required to restore unitarity. We find that `not-so-small' or `large' correction would be required to achieve a unitary evolution.

However, we find that such a condition leads to significant deviation from the established results by Page \cite{Page-1}, \cite{Page-2}, \cite{Page-3}. This leads us to conclude that the amount of correction dictated by our qubit formalism is not compatible with any physical unitary evolution. This result is in line with a more general result earlier proved in \cite{gidd-4} -- no order of correction in the form of admixture of Bell pair states suffices to retrieve information from black hole radiation.



The organization of the paper is as follows. Section \ref{sec-2} introduces the leading order formulation of the black hole information paradox. In section \ref{sec-3}, the authors present the motivation to look for resolving the paradox using `small' corrections. Section \ref{sec-4} discusses the arguments why the problem cannot be avoided incorporating `small' corrections. In section \ref{sec-5}, a simple toy model supporting the results of the previous section is briefly analyzed. Section \ref{sec-6} generalizes the arguments in section \ref{sec-4} and sections \ref{sec-7} establishes the insufficiency of Bell pair corrections for information retrieval from black hole radiation. In conclusion, section \ref{sec-9} summarizes the findings of the paper.

\section{Leading Order Formulation of the Black Hole Information Paradox}\label{sec-2}

In this section we take a quick review of the leading order formulation of the information paradox as presented in \cite{peda} and subsequently followed up in \cite{infall}, \cite{avery}, \cite{icece}, \cite{gidd-2}. The purpose is to familiarize the reader with the framework in which we shall pursue some aspects of possible resolution to the paradox in subsequent sections.
 
It can be showed \cite{hawk-3}, \cite{gidd-1} that the joint system of the particle pair created near the horizon can be given as --
\begin{equation}\label{pair}
\ket{\Psi}_{pair} = Ce^{\beta c^{\dagger} b^{\dagger}}\ket{0}_c \ket{0}_b
\end{equation}
where $\beta$ is a number of order unity, $c^\dagger$ and $b^\dagger$ are creation operators, $\ket{0}$ represents the vacuum state and $c,b$ represent the ingoing and outgoing quanta respectively. This state is generally an entangled one, indicated by a trivially non-vanishing  von Neumann entropy of any of the subsystems. Leading order analysis makes the following assumptions --

(a) The black hole geometry can be foliated by a set of spacelike slices continuous at the horizon \cite{exactly}, \cite{peda} . Evolution of states over these slices can be explained by local quantum field theory as long as black hole dimensions are much larger than Planckian order.

(b) The state of the collapsing shell is represented by the state vector $\ket{\Psi}_M$. During the evolution, state of this matter usually remains unchanged, i.e. $\ket{\Psi}_M \to \ket{\Psi}_M$. Given the slicing of spacetime as in (a) this assumption is inherently implied by locality. Large spacelike separations between the newly evolved pair and the collapsed matter prohibits any propagation of local effects.  However, allowing some local operator act on this matter subspace alone while leaving other subspaces unchanged will essentially produce the same result.

(c) State of the newly created pair can be approximated as
\begin{equation}\label{LO}
  \ket{\Psi}_{pair} = \frac{1}{\sqrt{2}}\left(\ket{0}_c\ket{0}_b + \ket{1}_c\ket{1}_b\right)
\end{equation}
This is one of the Bell pair states and hereafter we shall refer to it as the `Hawking state'. This state is a simplified version of (\ref{pair}) considering the newly evolved pair to be maximally entangled. It should be noted that instead of using an infinite dimensional Hilbert space as in (\ref{pair}), this simplified state reflects essentially the same physics on a much simpler qubit space. The entropy of entanglement associated to any of the subspaces is given by
\begin{equation}
  S_{ent} = \log 2
\end{equation}

(d) At each successive step of pair creation the initial matter and the earlier created quanta move along the spacelike slice by distances of the order of $R$, the radius of the black hole event horizon. Stretched spacelike slice causes a new pair to evolve according to (\ref{LO}) while the earlier qubits and the matter state moves farther along the spacelike slice. We ignore any influence of these quanta over the newly created pair and describe the entire system as a tensor product state-
\begin{equation}\
  \ket{\Psi} = \ket{\Psi}_M \otimes\ket{\Psi}_1 \otimes \ket{\Psi}_2 \otimes \ldots \otimes \ket{\Psi}_N
\end{equation}
where each of the $\ket{\Psi}_i$ state is given by (\ref{LO}). After emission of $N$ pairs, the entanglement entropy between the ingoing and outgoing quanta is given by
 \begin{equation}
  S_{ent} = N\log 2
\end{equation}

It is this monotonically increasing entangle entropy that lies at the heart of the paradox. If the remnant scenario is discarded as advocated in \cite{presk}, \cite{bek-3} and \cite{sussk}, the black hole completely vanishes and only the outgoing quanta remain as a thermal radiation with nonzero entropy despite starting with a pure state with zero entanglement entropy. 

One intriguing feature of this analysis is the increasing dimension of the Hilbert space. The black hole starts with a large but finite entropy and hence should be constrained to a quantum mechanical description within a Hilbert space of finite dimension \cite{gidd-3}. Unitary evolution preserves entropy which is in general a function of the dimension of the Hilbert space, though dimensions of individual subspaces can decrease as advocated in \cite{gidd-3}. This might intuitively imply an inevitable departure from unitarity. However, it has been showed in \cite{avery} that dimension can be effectively preserved in this framework as well. One should consider the additional $2k$ degrees of freedom after $k$ steps of pair production as auxiliary and not physical ones -- a mathematical tool introduced only to conveniently represent the phenomenon of pair production. Since we can always imagine an ancillary space to purify a mixed state by enlarging the Hilbert space, we might expect a restoration of unitarity over an enlarged Hilbert space provided that the radiation quanta at the end of the evaporation are in a pure state. Similar models involving nonlocal modes of information transfer have been presented in \cite{gidd-3}, \cite{gidd-4}.

\section{Motivation of Small Corrections to Leading Order}\label{sec-3}

The leading order analysis captures the coarse grained physical reality -- the salient features of the contradiction between black hole radiation and unitarity. However, it ignores any possible correlation between the outgoing quanta and the earlier ones or the black hole forming matter. Associating back reaction, quantum gravity effects or small perturbations to the Schwarzschild geometry might modify the state (\ref{LO}) of the particle pair. However, incorporating such details should bring about only some \emph{small} correction. This is because the horizon is still a low curvature region and LQFT should provide a reasonably accurate picture of quantum processes near the horizon.

Despite the predicted smallness of the correction, it has been widely speculated that these corrections might build up in course of time to eventually restore unitarity. Such speculations rely on the large number of radiated quanta over the lifetime of a black hole. One simple example might illustrate the effect of accumulating corrections in departure from predicted result. Let us assume that a particular process generates photons in a definite polarization state, say the horizontally polarized state --
\begin{equation}
\hat{\rho} = \ket{\rightarrow}\bra{\rightarrow}
\end{equation}
This is a pure state and hence possesses zero entanglement entropy. However, some error in the process might cause this state to be slightly modified as a mixed one.
\begin{equation}
\hat{\rho}' = \left(1 - \epsilon\right)\ket{\rightarrow}\bra{\rightarrow} + \epsilon \ket{\uparrow}\bra{\uparrow}
\end{equation}
Here $\epsilon$ is a small positive number, much less than unity. Entanglement entropy associated with this state is given by $S(\hat{\rho}') = -\left(\left(1-\epsilon\right)\log\left(1-\epsilon\right) + \epsilon\log\epsilon\right)$ which is close to zero. However, after $N$ such photons are emitted, entanglement entropy of the joint system could be considerably large where the entanglement entropy of the joint system of $N$ \emph{unmodified} photons will be still zero. This departure can also be quantified in measures of fidelity. As a measure of closeness of two quantum states, fidelity is defined as
\begin{equation}
F(\hat{\rho},\hat{\rho}') = \mathrm{Tr}\sqrt{\hat{\rho}^{1/2}\hat{\rho}'\hat{\rho}^{1/2}}
\end{equation}
$\hat{\rho}'$ closely resembles $\hat{\rho}$ as $F\left(\hat{\rho}, \hat{\rho}'\right)= \sqrt{1 - \epsilon}$. However,  we find the system of $N$ photons exhibit only low fidelity with the joint state of $N$ unmodified photons as $N$ becomes large.
\begin{equation}
F\left(\hat{\rho}^{\otimes N}, \hat{\rho}'^{\otimes N}\right) = \left(1-\epsilon\right)^{N/2} \to 0
\end{equation}
As accumulated small corrections over a large number of states can yield significant departure from a desired state, it is speculative to consider such processes in case of black hole radiation. In literature, alternative formulations of the radiation spectrum of black holes consider such corrections. For example, particle production has been associated to a tunneling picture in \cite{parikh}, \cite{parikh2}, \cite{tunnel} where existence of non trivial correlations among the radiated quanta have been demonstrated to exist because of strict imposition of conservation of energy. It has been advocated in \cite{zhang1} and \cite{zhang2} that non thermal corrections might lead to information leakage at least at late times. 

Mathur has addressed the issue of small corrections to leading order analysis in an abstract mathematical framework in \cite{peda}. In the next section we shall review the results from \cite{peda} that show that small corrections actually cannot restore unitarity in the context of black hole radiation.

\section{Small Correction to Hawking State - Mathur's Bound} \label{sec-4}
Let us assume that the created pair at each time-step of evolution is not invariably in the Hawking state, rather it can be in any state of the space spanned by the basis states
\begin{equation}
S^{(1)} = \frac{1}{\sqrt{2}} |0\rangle_{c_{n+1}} |0\rangle_{b_{n+1}} + \frac{1}{\sqrt{2}} |1\rangle_{c_{n+1}} |1\rangle_{b_{n+1}}
\end{equation}
and
\begin{equation}
S^{(2)} = \frac{1}{\sqrt{2}} |0\rangle_{c_{n+1}} |0\rangle_{b_{n+1}} - \frac{1}{\sqrt{2}} |1\rangle_{c_{n+1}} |1\rangle_{b_{n+1}}.
\end{equation}

Here we deliberately choose to avoid the subspace spanned by the states $|0\rangle_{c_{n+1}} |1\rangle_{b_{n+1}}$ and $|1\rangle_{c_{n+1}} |0\rangle_{b_{n+1}}$, because there is not much physical explanation for pair creation in such states. Moreover, a four dimensional space considering all these four states as basis states has been considered in \cite{infall}, and it shows no result essentially different from that obtained from a toy model in two dimensional analysis.

The complete system consists of the matter $M$, inside quanta $c_i$ and outside quanta $b_i$. Let us choose a basis $|\psi_i\rangle$ for the subsystem comprising matter $M$ and inside quanta $c_i$ and another basis $|\chi_i\rangle$ for the radiation subsystem comprising of the $b_i$ quanta. Then the state of the complete system can be expressed as
\begin{equation}
|\Psi_{M,c}, \psi_b(t_n)\rangle = \sum_{m,n}C_{m,n}\psi_m\chi_n.
\end{equation}

We can always perform Schmidt decomposition to express this state as
\begin{equation}
|\Psi_{M,c}, \psi_b(t_n)\rangle = \sum_{i}C_{i}\psi_i\chi_i.
\end{equation}

At next time-step of evolution, the $b_i$ quanta move farther apart from the vicinity of the hole. Since the hole can no longer influence their evolution, we consider that no further evolution takes place for the outgoing quanta. The created pair can be in a superposition of the states $S^{(1)}$ and $S^{(2)}$. Hence the state $\psi_i$ can evolve into
\[
\psi_i \rightarrow \psi_i^{(1)}S^{(1)} + \psi_i^{(2)}S^{(2)}
\]
where the state $\psi_i$ has been expressed as the tensor product of the state $\psi_i^{(i)}$ representing $\{M,c_i\}$ subsystem and $S^{(i)}$ representing the newly created pair. Since $S^{(1)}$ and $S^{(2)}$ are orthonormal states, unitarity requires that
\begin{equation}
\|\psi_i^{(1)}\|^2 + \|\psi_i^{(2)}\|^2 = 1.
\end{equation}
In leading order case, newly created pair is invariably in the state $S^{(1)}$; hence $\psi_i^{(1)} = \psi_i$ and $\psi_i^{(2)} = 0$.

Now,
\begin{eqnarray}
|\Psi_{M,c}, \psi_b(t_{n+1})\rangle &=& \sum_i C_i\big[\psi_i^{(1)}S^{(1)} + \psi_i^{(2)}S^{(2)}\big]\chi_i\\
&=&\Big[\sum_i C_i \psi_i^{(1)}\chi_i\Big]S^{(1)}+\Big[\sum_i C_i \psi_i^{(2)}\chi_i\Big]S^{(2)}\\
&=&\Lambda^{(1)}S^{(1)} + \Lambda^{(2)}S^{(2)}
\end{eqnarray}
where $\Lambda^{(1)}=\sum_i C_i \psi_i^{(1)}\chi_i$, $\Lambda^{(2)}=\sum_i C_i \psi_i^{(2)}\chi_i$.

Entanglement entropy of the $\{b\}$ quanta
\begin{eqnarray}
    S_{b_{n}} &=& - \mathrm{tr} \hat{\rho}_{b_{n}}\log\hat{\rho}_{b_{n}}\\ \nonumber
        &=& \sum_i |C_i|^2 \log |C_i|^2 = S_0.
\end{eqnarray}

Since earlier emitted outside quanta can no longer be influenced, we have the same entanglement entropy of $\{b\}$ quanta at time-step $t_{n+1}$.

Now, entanglement entropy of the pair $(b_{n+1},c_{n+1})$ with the rest of the system is given by
\begin{equation}
S(b_{n+1},c_{n+1})= - \mathrm{tr} \hat{\rho}_{b_{n+1},c_{n+1}} \log \hat{\rho}_{b_{n+1},c_{n+1}}.
\end{equation}
Density matrix for the system $(b_{n+1},c_{n+1})$ is
\begin{equation}
\hat{\rho}_{b_{n+1},c_{n+1}} = \left[
    \begin{array}{ccc}
    \langle\Lambda^{(1)}|\Lambda^{(1)}\rangle & \langle\Lambda^{(1)}|\Lambda^{(2)}\rangle\\
    \langle\Lambda^{(2)}|\Lambda^{(1)}\rangle & \langle\Lambda^{(2)}|\Lambda^{(2)}\rangle
    \end{array}
    \right].
\end{equation}
Again, normalization of $|\Psi_{M,c}, \psi_b (t_{n+1})\rangle$ requires
\[ \|\Lambda^{(1)}\|^2 + \|\Lambda^{(2)}\|^2 = 1. \]

Mathur defined the correction to leading order Hawking state to be `small' in the sense that $\|\Lambda^{(2)}\|^2 <\epsilon$, where $\epsilon \ll 1$ \cite{peda}. This definition implies that there is very small admixture of the $S^{(2)}$ state with the $S^{(1)}$ state when new particle pairs are generated. Under such `small' departure from leading order semi-classical Hawking analysis, Mathur showed that entanglement entropy at each time-step increases by at least $\log 2 - 2\epsilon$ \cite{peda}. Since $\epsilon$ is a very small number, by definition, there is still order unity increase in entanglement entropy at each time-step, when `small' corrections are allowed. This result has been exemplified by a simple model incorporating small correlations between quanta created in consecutive steps \cite{infall}. Some other toy models of evaporation have been studied in \cite{mathur-plumberg}, which also conform to this result. In the next section we shall illustrate Mathur's bound in a toy model that incorporates a more general correlation compared to the model presented in \cite{infall}.

\section{A Simple Toy Model Incorporating Small Correction} \label{sec-5}
Let us consider a model of evaporation that allows all previously radiated quanta to modify the state of the newly evolved pair. At the first stage of evaporation, the state of the pair is given by the Bell pair state
\begin{equation}
\ket{\Psi_1} = \frac{1}{\sqrt{2}} \left( \ket{0}_b \ket{0}_c + \ket{1}_b \ket{1}_c \right)
\end{equation}
Assuming that the Hilbert space of the newly evolved pair is spanned by the $\ket{00}$ and $\ket{11}$ states only, the general form of the $n$ pair state is given by
\begin{equation}
\ket{\Psi_n} = \sum_{i=0}^{2^n - 1} a_i \ket{i}_b \ket{i}_c
\end{equation}
Here $\ket{i}$ denotes the $n$ bit representation of the integer $i$. Normalization requires that $\sum_{i=0}^{2^n - 1} |a_i|^2 = 1$. Let the state of the newly evolved pair depends on the $n$ qubit state $\ket{i}$ and the evolution is given by
\begin{equation}
\ket{\Psi_{n+1}} = \sum_{i=0}^{2^n - 1} a_i \ket{i}_b \ket{i}_c \otimes \left( \frac{e^{s_{i,n,0}}}{\sqrt{2}} \ket{0}_b \ket{0}_c + \frac{e^{s_{i,n,1}}}{\sqrt{2}} \ket{1}_b \ket{1}_c \right)
\end{equation}
The term $\frac{e^{s_{i,n,j}}}{\sqrt{2}}$ denotes the probability amplitude of observing the new pair at the state $\ket{j}_b \ket{j}_c$ given the joint state of the earlier pairs is given by $\ket{i}_b \ket{i}_c$. If the correction introduced is small, $\left| s_{i,n,j} \right|$ has to be a small positive number. Normalization requires
\begin{equation}\label{norm}
\sum_{j=0}^{1} e^{2s_{i,n,j}} = 2
\end{equation}
Entanglement entropy of the $n+1$ radiated quanta is hence given by
\begin{align*}
S(n+1) &= - \sum_{i} \sum_{j} \left( \frac{a_i e^{s_{i,n,j}}}{\sqrt{2}} \right)^2 \log \left( \frac{a_i e^{s_{i,n,j}}}{\sqrt{2}} \right)^2 \\
& = - \sum_{i} \sum_{j} a_i^2 e^{2s_{i,n,j}} \left( \log a_i + s_{i,n,j} - \frac{1}{2}\log 2 \right) \\
& = -\sum_i a_i^2 \log a_i^2 + \log 2 - \sum_i a_i^2 \sum_j s_{i,n,j} e^{2s_{i,n,j}}
\end{align*}
Identifying the first term in the last line as the entanglement entropy of the first $n$ quanta and defining $\Delta S = S(n+1) - S(n)$
\begin{equation}\label{ds}
\Delta S = \log 2 - \sum_i a_i^2 \sum_j s_{i,n,j} e^{2s_{i,n,j}}
\end{equation}
Since $s_{i,n,j}$ are quantities of small magnitude, using the approximation $e^x \sim 1 + x$ in (\ref{norm})
\begin{equation}
\sum_{j=0}^{1} (1 + 2s_{i,n,j}) = 2 \Rightarrow s_{i,n,0} = -s_{i,n,1}
\end{equation}
Using these smallness approximations in (\ref{ds})
\begin{align*}
\Delta S &= \log 2 - \sum_i a_i^2 \sum_j s_{i,n,j} (1 + 2s_{i,n,j}) \\
& = \log 2 - 4\sum_i a_i^2 s_{i,n,0}^2 \\
& \ge \log 2 - 4 \mathrm{max}(s_{i,n,0}^2)
\end{align*}
which implies that $\Delta S$ is a necessarily positive quantity as long as the smallness condition holds.

\section{Generalization of Mathur's Bound} \label{sec-6}

In this section we review the generalization of Mathur's bound as presented in \cite{icece}. This generalization is motivated with a view to exploring the effects of introducing corrections of arbitrary magnitude to the leading order analysis and explore, if any, the possibilities of restoration of unitarity. The results in \cite{icece} establish nontrivially strong upper and lower bounds to $\Delta S$. To facilitate the derivation, we first derive two lemmas leading to the derivation of the final result as a theorem.

Let us introduce the correction parameters and corresponding quantities first:
\begin{equation}
\langle\Lambda^{(2)}|\Lambda^{(2)}\rangle = \epsilon^2,
\end{equation}
\begin{equation}
\langle\Lambda^{(1)}|\Lambda^{(1)}\rangle = 1-\epsilon^2,
\end{equation}
\begin{align}\label{eqn:con3}
\langle\Lambda^{(1)}|\Lambda^{(2)}\rangle = \langle\Lambda^{(2)}|\Lambda^{(1)}\rangle = \epsilon_2,
\end{align}
\begin{equation}
\gamma^2 = 1-4\big[\epsilon^2(1-\epsilon^2)-\epsilon_2^2\big].
\end{equation}

We shall call the correction small if $|\epsilon| \ll 1$. It should be noted that $\epsilon_2$ should in general be a complex quantity and $\langle\Lambda^{(1)}|\Lambda^{(2)}\rangle = \langle\Lambda^{(2)}|\Lambda^{(1)}\rangle^{*}$. However, in (\ref{eqn:con3}) we assume a real $\epsilon_2$ for the sake of simplicity, though letting it have complex values will give the same result.

\newtheorem{lone}{Lemma}
\begin{lone}
Entanglement entropy of the newly created pair is given by \[S(p)\leq\sqrt{1-\gamma^2}\log2.\]
\end{lone}

\begin{proof}
Reduced density matrix for the pair
\begin{equation}
\hat{\rho}_p = \left[
    \begin{array}{ccc}
    1-\epsilon^2 & \epsilon_2\\
    \epsilon_2 & \epsilon^2
    \end{array}
    \right].
\end{equation}

Eigenvalues of this matrix are: $\lambda_1=\frac{1+\gamma}{2}$ and $\lambda_2=\frac{1-\gamma}{2}$.
Hence entanglement entropy of the pair is

\begin{eqnarray}
S(p) &=& -\mathrm{tr} \hat{\rho}_p \log \hat{\rho}_p = -\sum_{i=1}^2 \lambda_i \log \lambda_i \nonumber \\
    &=& \log 2 - \frac{1}{2}[(1+\gamma)\log(1+\gamma)+ (1-\gamma)\log(1-\gamma)].
\end{eqnarray}

It can be shown easily for $0 \leq x \leq 1$ that
\begin{eqnarray}\label{label14}
(1-x^2)\log2 &\leq& \log2 - \frac{1}{2}[(1+x)\log(1+x)\\ \nonumber
&&+(1-x)\log(1-x)] \leq \sqrt{1-x^2}\log2.
\end{eqnarray}
Now, the result follows from (\ref{label14}).
\end{proof}

\newtheorem{lthree}[lone]{Lemma}
\begin{lthree}
\begin{equation}
(1-4\epsilon_2^2)\log 2 \leq S(b_{n+1})=S(c_{n+1})\!\leq\! \sqrt{1-4\epsilon_2^2}\log2 \nonumber
\end{equation}
\end{lthree}

\begin{proof}
The complete state of the system after creation of $n+1$ pairs
\begin{eqnarray}
|\Psi_{M,c}, \psi_b(t_{n+1})\rangle  &=& \Big[|0\rangle_{c_{n+1}}|0\rangle_{b_{n+1}}\frac{1}{\sqrt{2}}(\Lambda^{(1)}+\Lambda^{(2)})\Big]\nonumber\\
&&+\Big[|1\rangle_{c_{n+1}}|1\rangle_{b_{n+1}}\frac{1}{\sqrt{2}}\!(\!\Lambda^{(1)}-\Lambda^{(2)})\Big].
\end{eqnarray}
Now, the reduced density matrix describing $c_{n+1}$ or $b_{n+1}$ quanta is
\begin{eqnarray}
\hat{\rho}_{b_{n+1}} &=& \hat{\rho}_{c_{n+1}}\\
&=&\left[\begin{array}{cc}
\frac{1}{2}\langle(\Lambda^{(1)}\!+\Lambda^{(2)})|(\Lambda^{(1)}+\Lambda^{(2)})\rangle & 0 \\
0 & \frac{1}{2}\langle(\Lambda^{(1)}-\Lambda^{(2)})|(\Lambda^{(1)}-\Lambda^{(2)})\rangle
\end{array}\right]\nonumber\\
&=& \left[\begin{array}{cc}
\frac{1+2\epsilon_2}{2} & 0\\
0 & \frac{1-2\epsilon_2}{2}
\end{array}\right].
\end{eqnarray}
Then, entanglement entropy of the $c_{n+1}$ or $b_{n+1}$ quanta is
\begin{eqnarray}
S(b_{n+1})=S(c_{n+1}) &=& \log2 - \frac{1+2\epsilon_2}{2}\log(1+2\epsilon_2)\nonumber\\
        &&-\: \frac{1-2\epsilon_2}{2}\log(1-2\epsilon_2).
\end{eqnarray}
Now the result follows directly from (\ref{label14}).
\end{proof}

With the use of these lemmas, we now prove our desired result as a theorem.
\vspace{3 mm}
\newtheorem{nthm}{Theorem}
\begin{nthm}
Change of entanglement entropy from time-step $t_n$ to $t_{n+1}$ is restricted by the following bound:
\begin{equation}\label{label10}
1-4\epsilon_2^2-\sqrt{1-\gamma^2} \leq \frac{\Delta S}{\log 2} \leq \sqrt{1-4\epsilon_2^2}
\end{equation}
where $\Delta S = S(b_{n+1}, \{b\})-S(\{b\})$.
\end{nthm}

\vspace{2 mm}
\begin{proof}
Let us assume $A=\{b\}, B=b_{n+1}, C=c_{n+1}$.\\
Using \emph{strong subadditivity inequality} we have,
\[S(A)+S(C) \leq S(A,B)+S(B,C)\]
\[\Rightarrow S(\{b\})+S(c_{n+1})\leq S(\{b\},b_{n+1})+S(b_{n+1},c_{n+1})\]
\begin{equation}\label{label8}
\Rightarrow \Delta S \geq (1-4\epsilon_2^2)\log 2 - \sqrt{1-\gamma^2}\log 2.
\end{equation}
Inequality (\ref{label8}) follows from using Lemma 1 and Lemma 2.

Now using \emph{subadditivity inequality} we have,
\[S(A)+S(B) \geq S(A,B)\]
\[\Rightarrow S(\{b\})+S(b_{n+1}) \geq S(\{b\},b_{n+1})\]
\begin{equation}\label{label9}
\Rightarrow \Delta S \leq \sqrt{1-4\epsilon_2^2}\log 2
\end{equation}
Inequality (\ref{label9}) follows from Lemma 2.

The result follows from combining (\ref{label8}) and (\ref{label9}).
\end{proof}
\vspace{2 mm}

 In establishing the upper and lower bounds of change in entanglement entropy, we have dropped the concept of `smallness' of the parameters opening the window for a wide range of speculated effects that may cause departure from the leading order analysis of black hole evaporation. Implication of the result (\ref{label10}) is twofold -- (a) in putting an upper bound on the change in entanglement entropy and (b) in proving a quantitative measure of `correction' that needs to be incorporated in any hope to restore unitarity in the process of Hawking radiation.



\subsection{Nontriviality}
The result obtained in \cite{peda} and its subsequent generalization in \cite{avery} do not engage in any argument that address large corrections. As a result, both results perform poorly asymptotically. We can establish the nontriviality of (\ref{label10}) by looking at its limiting characteristics, at $\epsilon \to 0$ and $\epsilon \to 1$. In both cases, it is straightforward to calculate that both limits approach unity representing maximal increase of entropy, as expected from the model considered. Nontriviality may further be inferred from the maximum difference of the lower and upper bound in (\ref{label10}), as we expect this difference to never exceed $2\log 2$. Denoting this quantity as $D_{\Delta S}$,
\begin{equation}\label{dds}
D_{\Delta S} = \log 2 \left(\sqrt{1-4\epsilon_2^2} + \sqrt{4\left(\epsilon^2(1-\epsilon^2) - \epsilon_2^2\right)} - (1-4\epsilon_2^2)\right)
\end{equation}
It is straightforward to note that for any value of $\epsilon_2$, $D_{\Delta S}$ is maximized when $\epsilon^2(1-\epsilon^2)$ is maximum i.e. $\frac{1}{4}$. Then (\ref{dds}) reduces to
\begin{equation}
\frac{D_{\Delta S}}{\log 2} = 2\sqrt{1-4\epsilon_2^2} - (1-4\epsilon_2^2)
\end{equation}
It is straightforward to show that maximization of $D_{\Delta S}$ requires $\frac{dD_{\Delta S}}{d\epsilon_2} = 0 \Rightarrow \epsilon_2 = 0$ implying
\begin{equation}
\frac{D_{\Delta S}}{\log 2} \le 1
\end{equation}
and hence establishing nontriviality of the bound as expected.



\subsection{Small corrections are not enough}
    The result in (\ref{label10}) facilitates us to demonstrate that small corrections do not suffice to restore unitarity in the process of Hawking radiation. Retrieval of information requires monotonic decrease of entanglement entropy after half of the black hole lifetime \cite{Page-1}. This implies a necessarily negative lower bound in (\ref{label10}) implying


    \begin{equation}\label{label11}
    4\epsilon_2^2 (1-4\epsilon_2^2) > 1-4\epsilon^2 (1-\epsilon^2)
    \end{equation}
    after replacing $\gamma^2 = 1-4\big[\epsilon^2(1-\epsilon^2)-\epsilon_2^2\big]$ and some manipulation.

   The quantity in the left in (\ref{label11}) is bounded from above by a maximum value of $\frac{1}{4}$ implying

    \[1-4\epsilon^2 (1-\epsilon^2) < \frac{1}{4}\]
    \begin{equation}\label{condn}
    \Rightarrow \frac{1}{2}<\epsilon < \frac{\sqrt{3}}{2}.
    \end{equation}

This bound on $\epsilon$ gives us the necessary (not sufficient) order of correction for any possible retrieval of quantum information from black hole radiation. The bound in (\ref{condn}) implies that the cherished order of the correction parameter has to be far greater than what can be dubbed as some \emph{small correction}. While conforming to the results obtained in \cite{peda} and \cite{avery}, the derivation of the bound in (\ref{label10}) does not make any presumption on the order of the correction parameter. The fact that the inequality in (\ref{label10}) remains valid for any order of correction magnitude is a unique feature of the derivation presented in this paper and will contribute to another significant result in the following section. 

As anticipated in \cite{icece}, the inequality in (\ref{condn}) can be interpreted as a necessary condition for information retrieval from black hole radiation. Interestingly, generalizing the model presented by Mathur, Giddings and Shi have shown that no order of correction is sufficient to restore unitarity if admixture of Bell pair states only is allowed \cite{gidd-4}. The following section looks into the issue in light of the qubit model of evaporation we have been considering so far.

\section{Incompatibility of Bell Pair State Correction to Page Curves}\label{sec-7}

In his seminal paper \cite{Page-1}, Don Page conjectured his celebrated result about the average entropy of a subsystem which was later proved by a number of authors in \cite{Pageproof-1}, \cite{Pageproof-2}, \cite{Pageproof-3}. The central idea of this paper was calculating the average behavior of a subsystem of random pure state in a finite dimensional Hilbert space. Assuming that the pure state lives in a Hilbert space of dimension $mn$ where $m$ represents the dimension of the smaller subsystem, the entanglement entropy, denoted by $S_{m,n}$, is given by

\begin{equation}
S_{m,n} = \log m - \frac{m}{2n}
\end{equation}

This result has a simple interpretation, a smaller subsystem will appear to be maximally entangled with the rest of system. Page extrapolated his results for black hole radiation and approximated the time evolution of the radiation entropy for a typical black hole in a pure state \cite{Page-2}, \cite{Page-3}. As shown in figure \ref{fig1}, radiation entropy keeps rising till the halfway radiation point, by the time almost half of the black hole mass has been radiated away. Then the radiation entropy starts decreasing and returns to zero at the end of the black hole lifetime. This point of reversal for radiation entropy implies that actual quantum information is revealed by black hole radiation. Based on the results by Page, we conclude that any unitarity respecting evaporation model of a black hole should closely proximate the Page curve. We shall investigate the behavior of radiation entropy for the model that we consider and eventually show that if our model implies unitarity, evolution of entanglement entropy has to deviate largely from the Page curve. This would lead us to conclude that no order of correction in the model that we consider can actually imply unitarity.

\begin{figure}
  \centering
  \includegraphics[width=3in]{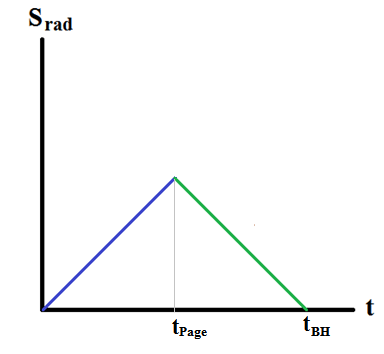}\\
  \caption{The Page curve for a radiating system}\label{fig1}
\end{figure}


    Our result in (\ref{label10}) generates the envelope of any Page-like curve for the model of black holes we have been considering. The monotonically increasing part is bounded by a straight line with a slope equal to the maximum of the upper bound in (\ref{label10}) which is equal to $\log 2$. Similarly, the decreasing part is also bounded by a straight line and the slope of this line is obtained from the minimum of the lower bound in (\ref{label10}). \\

      \begin{eqnarray*}
         \text{The lower bound} &=& 1 - 4\epsilon_2^2 - \sqrt{1-\gamma^2} \\
          &=& 1 - 4\epsilon_2^2 - \sqrt{4\epsilon^2 (1 - \epsilon^2) - 4\epsilon_2^2}
      \end{eqnarray*}
    For any value of $\epsilon_2$, this quantity is minimized when $\epsilon^2(1-\epsilon^2)$ is maximized, i.e., $\epsilon^2(1-\epsilon^2) = \frac{1}{4}$. Therefore, we need to consider the minimum value of the quantity $1 - 4\epsilon_2^2 - \sqrt{1 - 4\epsilon_2^2}$, that happens to be $-\frac{1}{4}$. The Page-like curve that our evaporation model generates has to be bounded in the grey region in figure \ref{fig2}.

\begin{figure}
  \centering
  \includegraphics[width=3in]{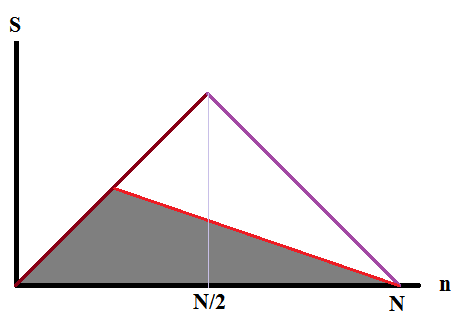}\\
  \caption{The Page-like curves}\label{fig2}
\end{figure}

  Entanglement between the radiation and the remaining black hole systems governs the magnitude of radiation entropy and this is dominated by the dimension of the smaller subsystem. For the bounded region in figure \ref{fig2}, the late radiation entropy deviates significantly from maximal entanglement. Hence, if the evaporation model based on generation of Bell pair states has to be unitary, it would demand significant deviation from the well-established results by Page. This is `problematic', because a generic subsystem should not be subject to such restrictions.

  The qubit evaporation model in question relies on the assumption that evaporation is governed by generation of Bell pair states. However, it keeps the possibility of `any' required amount of admixture open. We find that this generalization does not help to achieve a `physically meaningful' unitary evaporation that conforms to the Page curve. This is in agreement with the result in \cite{gidd-4} that \emph{no order of correction in admixture with Bell pair states can restore unitarity in the process of Hawking radiation}.

\section{Conclusion} \label{sec-9}
Mathur's results clearly demonstrate that `small' corrections to the leading order Hawking analysis does not unitarize the black hole evaporation process. In this paper, we reassure this by analyzing a toy qubit model that incorporates a generalized quantum correlation between successive pairs. We then generalize Mathur's work by relaxing the `smallness' condition. We establish a rigorous nontrivial upper and lower bound for change in entanglement entropy, that helps us to parameterize the amount of correction to Hawking state in the form of one of the Bell pair states. Our results show that the correction required to restore unitarity in the evaporation process is `not-so-small' or even `large'. However, even if we allow such an evaporation law, we find that it is at odds with the predictions of Page curve. If we assume that our evolution law is unitary, it should not deviate much from the Page curve. Thus this leads us to reinforce the recent result that information paradox cannot be resolved if we stick to the picture of pair generation in Bell states.
\section*{Acknowledgements}

The authors gratefully acknowledge Mahbub Majumdar for his insightful directions and useful discussions in different steps of this work. A part of this work was done during the workshop on black hole information paradox at Harish-Chandra Research Institute, Allahabad, India in February 2014. AR would like to thank the participants of the workshop for useful discussions and helpful comments.

\bibliographystyle{unsrt}

\bibliography{mathur}

\end{document}